\documentclass[conference]{IEEEtran}
\makeatletter
\def\ps@headings{%
\def\@oddhead{\mbox{}\scriptsize\rightmark \hfil \thepage}%
\def\@evenhead{\scriptsize\thepage \hfil \leftmark\mbox{}}%
\def\@oddfoot{}%
\def\@evenfoot{}}
\makeatother
\usepackage{amsthm,cite}
\usepackage{amssymb, amsmath, color}
\usepackage{nicefrac}
\usepackage{subfigure}
\usepackage{graphicx}
\usepackage{url}

\newcommand{\expec}{\mathbb{E}}

\newtheorem{lem}{Lemma}
\newtheorem{thm}{Theorem}
\newtheorem{defn}{Definition}

\newtheorem{conj}{Conjecture}
\newtheorem{rem}{Remark}
\newtheorem{exmp}{Example}

\title{The Effect of Block-wise Feedback on the Throughput-Delay Trade-off in Streaming}

\author{
\authorblockN{Gauri Joshi}
\authorblockA{EECS Dept.,
MIT\\
Cambridge, MA 02139, USA \\
Email: gauri@mit.edu}
\and
\authorblockN{Yuval Kochman}
\authorblockA{School of CSE, HUJI \\
Jerusalem, Israel \\
Email: yuvalko@cs.huji.ac.il}
\and
\authorblockN{Gregory W. Wornell}
\authorblockA{EECS Dept.,
MIT\\
Cambridge, MA 02139, USA \\
Email: gww@mit.edu}
}

\begin{document}

\sloppy

\maketitle
\begin{abstract}
Unlike traditional file transfer where only total delay matters, streaming applications impose delay constraints on each packet and require them to be \emph{in order}. To achieve fast in-order packet decoding, we have to compromise on the throughput. We study this trade-off between throughput and in-order decoding delay, and in particular how it is affected by the frequency of block-wise feedback to the source. When there is immediate feedback, we can achieve the optimal throughput and delay simultaneously. But as the feedback delay increases, we have to compromise on at least one of these metrics. We present a spectrum of coding schemes that span different points on the throughput-delay trade-off. Depending upon the delay-sensitivity and bandwidth limitations of the application, one can choose an appropriate operating point on this trade-off. %The proposed schemes are optimal the no feedback, and with small feedback delay cases, and close to optimal in general.  
\end{abstract}
%\newpage

\section{Introduction}
\label{sec:intro}

\subsection{Motivation}
A recent report \cite{sandvine_report} shows that $62 \%$ of the Internet traffic in North America comes from real-time streaming applications such as NetFlix ($28.88\%$) and YouTube ($15.43 \%$). Streaming traffic consumes such a large fraction of Internet bandwidth because video files inherently have a larger size than other forms of data. Thus, there is a need to develop transmission schemes which can ensure a high quality of experience to the user, with efficient use of available bandwidth. 

Unlike traditional file transfer where only total delay matters, streaming imposes delay constraints on each individual packet. Further, many applications require in-order playback of packets at the receiver. Packets received out of order are buffered until the missing packets in the sequence are successfully decoded. In audio and video applications some packets can be dropped without affecting the streaming quality. However, other applications such as remote desktop, and collaborative tools such as Dropbox and Google Docs have strict order constraints on packets, where packets represent instructions that need to be executed in order at the receiver. 

%\subsection{Main Idea}
To ensure that packets are decoded in order, the transmission scheme must give higher priority to older packets that were delayed, or received in error due to channel noise. However, repeating old packets instead of transmitting new packets results in a loss in the overall rate at which packets are delivered to the user, that is, the throughput. Thus there is a fundamental trade-off between throughput and in-order decoding delay. 

The throughput loss incurred to achieve low in-order decoding delay can be significantly reduced if the source receives feedback about packet losses, and thus can adapt its future transmission strategy to strike the right balance between old and new packets. We study this interplay between feedback and the throughput-delay trade-off. %This analysis can help design transmission schemes that achieve the best throughput-delay trade-off with a limited amount of feedback. 

%In the ideal case of immediate and error-free feedback to the source, Automatic-repeat-request (ARQ) protocols where the source only retransmits lost packets are optimal under any delay metric. However, when feedback is lossy, delayed or completely absent, the encoder has inherent uncertainty about the state of the decoder. It must strike a balance between transmitting new packets and repeating old packets that could have been erased.

\subsection{Previous Work}

Only a few papers in literature have analyzed streaming codes. Fountain codes \cite{luby} are capacity-achieving erasure codes, but they are not suitable for streaming because the decoding delay is proportional to the size of the data. Streaming codes without feedback for constrained channels such as adversarial and cyclic burst erasure channels were first proposed in \cite{emin_thesis}, and also extensive explored in \cite{badr_khisti_isit, patil_badr_khisti}. The thesis \cite{emin_thesis} also proposed codes for more general erasure models and analyzed their decoding delay. Decoding delay has also been analyzed studied in \cite{arq_shah_medard, delay_control_online_nw_coding} in a multicast scenario with immediate feedback to the source.

%These codes are based upon sending linear combinations of source packets; indeed, it can be shown that there is no loss in restricting the codes to be linear. 
However, decoding delay does not capture \emph{in order } packet delivery which is required for streaming applications. This aspect is captured in the delay metrics in \cite{huan_paper} and \cite{gauri_isit_paper}, which consider that packets are played in-order at the receiver. The authors in \cite{huan_paper} analyze the throughput-delay trade-off for uncoded packet transmission over a channel with long feedback delay. In \cite{gauri_isit_paper} we propose coding schemes that minimize playback delay in point-to-point streaming for the no feedback and immediate feedback cases. However, the case of block-wise feedback to the source remains to be explored. 

\subsection{Our Contributions}
In this paper we consider this unexplored problem of how to effectively utilize block-wise feedback to the source to ensure in-order packet delivery to the user. In contrast to playback delay considered in \cite{huan_paper} and \cite{gauri_isit_paper}, we propose a more versatile delay metric called the in-order decoding exponent. This metric captures the burstiness in the in-order decoding of packet for applications which require packets in-order, but do not necessarily play them at a constant rate.

When there is immediate feedback, we can achieve the best throughput-delay trade-off. But when the feedback comes in blocks, we have to compromise on the throughput to ensure fast in-order decoding. We present a spectrum of coding schemes that span different points on the throughput-delay trade-off. Depending upon the delay-sensitivity, and bandwidth limitations of the application, one can choose an appropriate operating point on this trade-off. The proposed codes can be shown to be optimal over a broad class of schemes for the no feedback, and small feedback delay cases. 

\section{Problem Setup}
\label{sec:prob_setup}

%\begin{figure*}
%\centering
%\includegraphics[width=6.5in]{sys_model_fb.eps}
%\caption{The system model for streaming with delayed feedback. In slot $k$, source can transmit packets based on feedback about channel erasures up to slot $k-d$. If the channel is erased $E_k = 0$, otherwise $E_k = 1$.\label{fig:sys_model_fb}}
%\vspace{-0.3cm}
%\end{figure*} 

\subsection{System Model}
We consider a point-to-point packet streaming scenario where the source has a large stream of packets $s_1, s_2, \cdots, s_n$. The encoder creates a coded packet $y_{n} = f(s_1,\,\,s_2 \,\,..s_n)$ in each slot $n$ and transmits it over the channel. The encoding function $f$ is known to the receiver. For example, if $y_{n}$ is a linear combination of the source packets, the coefficients are included in the transmitted packet so that the receiver can use them to decode the source packets from the coded combination. Without loss of generality, we can assume that $y_n$ is a linear combination of the source packets. 

We consider an i.i.d.\ packet erasure channel where every transmitted packet is correctly received with probability $p$, and otherwise received in error and discarded. An erasure channel is a good model when encoded packets have a set of checksum bits that can be used to verify with high probability whether the received packet is error-free. 
%Such applications audio and video streaming, as well as remote desktop, and collaborative tools such as Dropbox, Google docs. In the latter, packets represent instructions that need to executed in-order at the receiver. 

The receiver application requires the stream of packets to be \emph{in order}. Packets received out of order are buffered until the missing packets in the sequence are decoded. Due to this in-order property, the transmitter can stop including $s_k$ in coded packets when it knows that the receiver can decode $s_k$ once all $s_i$ for $i < k$ are decoded. We refer such packets as ``seen" packets. The notion of ``seen" is defined formally as follows. %when it is ``seen". The notion of a packet being ``seen" is defined as dolin coded combinations if it can be decoded once all packets before it are decoded. We define a notion called ``seen" to identify such packets. 
%We assume that the buffer at the receiver is large enough so that it does not overflow with high probability. 
%\subsection{Feedback Model}

\begin{defn}[Seen Packets]
\label{defn:seen_pkts}
A packet $s_k$ is said to be ``seen" by the transmitter when it knows that a coded combination that only includes $s_k$ and packets $s_i$ for $1 \leq i \leq k$ is received successfully. %Once marked as ``seen" the packet need not be included in any future coded packets. 
\end{defn}

We consider that the source receives block-wise feedback about channel erasures after every $d$ slots. Thus, before transmitting in slot $kd+1$, for all integers $k \geq 1$, the source knows about the erasures in slots $(k-1)d +1$ to $kd$. It can use this information to adapt its transmission strategy in slot $kd +1$. Block-wise feedback can be used to model a half-duplex communication channel where after every $d$ slots of packet transmission, the channel is reserved for the receiver to send feedback about the status of decoding. Note that the feedback can be used to estimate $p$, the probability of success of the erasure channel, when it is unknown to the source.

%Note that the case $d= 1$, corresponds to immediate feedback when the source has complete knowledge about past erasures. %And when $d \rightarrow \infty$, the block-wise feedback model converges to the scenario where there is no feedback to the source. 
%
%The second type of feedback is delayed feedback where the source has complete knowledge of erasures up to slot $n-d$ before transmitting in slot $n$. This type of feedback can be used to model a transmission delay on the channel which is longer than the length of each slot. For example, the one-way transmission delay is $l$ slots, then the overall feedback delay $d = 2l$ slots. 
%
%Note that in the special case $d= 1$, which corresponds to immediate feedback case with $d= 1$, the block-wise and delayed feedback models are identical. And when $d \rightarrow \infty$, both models converge to the scenario where there is no feedback to the source. With both types of feedback, the information about past erasures can be used to estimate $p$, the probability of success of the erasure channel. Thus, the coding schemes we propose are universal -- that is they can be used even when the channel quality of unknown to the source.

\subsection{Throughput and Delay Metrics}
%The receiver requires packets to be in-order. Thus a packet $s_n$ is said to be decoded only when all past packets are also decoded. As a result, we can define decoding instants when bursts of successive packets are decoded. 
We consider two metrics to measure the quality of streaming, the throughput $\tau$ and in-order decoding exponent $\lambda$. The throughput is the rate at which ``innovative" coded packets are received. A coded packet is said to be ``innovative" if it is linear independent with respect to the coded packets received until then. The bandwidth required is proportional to $1/ \tau$. 
%
%The receiver application may require a minimum level of throughput. For example, if applications with playback require $\tau$ to be greater than the playback rate. 
%
The throughput captures the overall rate at which packets go through the channel, irrespective of the order. The \emph{in-order} decoding aspect is captured by a metric called the in-order decoding exponent $\lambda$ which is defined as follows.

\begin{defn}[In-order Decoding Exponent]
\label{defn:lambda_def}
Let $T$ be the time between two successive instants of decoding one or more packets in-order. Then the in-order decoding exponent $\lambda$ is
\begin{equation}
\lambda \triangleq -\lim_{n \rightarrow \infty} \frac{\log \Pr(T>n)}{n}. \label{eqn:lambda_def}
\end{equation}
\end{defn}
The relation \eqref{eqn:lambda_def} can also be stated as $\Pr(T>n) \doteq e^{-n \lambda}$ where $\doteq$ stands for asymptotic equality defined in \cite[Page 63]{thomas_cover}. The in-order decoding exponent captures the burstiness in packet decoding. For example, if the streaming application plays one in-order packet in every slot, and if there are $b$ packets in the receiver buffer, then the probability of an interruption in playback is proportional to $e^{-\lambda b}$.

In this paper we analyze how the trade-off between $\tau$ and $\lambda$ is affected by the block-wise feedback delay $d$. We first consider the extreme cases of immediate feedback $(d=1)$ and no feedback $(d=\infty)$ in Section~\ref{sec:immediate_feedback} and Section~\ref{sec:no_feedback} respectively. This gives us insights into the analysis of the $(\tau,\lambda)$ trade-off for general $d$ in Section~\ref{sec:block_wise_fb}. 

\section{Immediate Feedback}
\label{sec:immediate_feedback} 

In the immediate feedback $(d=1)$ case, the source has complete knowledge of past erasures before transmitting each packet. We can show that a simple automatic-repeat-request (ARQ) scheme is optimal in both $\tau$ and $\lambda$. In this scheme, the source transmits the lowest index unseen packet, and repeats it until the packet successfully goes through the channel. 

Since a new packet is received in every successful slot, the throughput $\tau =p$, the success probability of the erasure channel. The ARQ scheme is throughput-optimal because the throughput $\tau = p$ is equal to the information-theoretic capacity of the erasure channel \cite{thomas_cover}. Moreover, it also gives the optimal the in-order decoding exponent $\lambda$ because one in-order packet is decoded in every successful slot. To find $\lambda$, first observe that the tail distribution of the time $T$, the interval between successive in-order decoding instants is,
\begin{align}
\Pr(T>n) &= (1-p)^{n} 
\end{align}
Substituting this in Definition~\ref{defn:lambda_def} we get the exponent $\lambda = -\log(1-p)$. Based on this analysis of the immediate feedback case, we can find limits on the range of achievable $(\tau,\lambda)$ for any feedback delay $d$ as follows.
\begin{lem}
\label{lem:range_of_tau_lambda}
The throughput and delay metrics $(\tau,\lambda)$ achievable for any feedback delay $d$ lie in the region $0 \leq \tau \leq \rho$, and $0 \leq \lambda \leq -\log(1-p)$. 
\end{lem}
\begin{proof}
When feedback is received after blocks of $d > 1$ slots, the source has less knowledge about past erasures than in the immediate feedback ($d=1$) case. Thus, the $(\tau, \lambda)$ trade-off when $d > 1$ is always worse than $(\tau, \lambda) = (p,-\log(1-p))$ the optimal trade-off for the immediate feedback ($d=1$) case.
\end{proof}

\section{No Feedback}
\label{sec:no_feedback}
Now we consider the other extreme case $(d = \infty)$, where there is no feedback to the source about channel erasures. We propose a coding scheme and prove that it gives the best $(\tau, \lambda)$ trade-off among a class of codes called full-rank codes which are defined as follows. 
\begin{defn}[Full-rank Codes]
\label{defn:full_rank_codes}
In slot $n$ we transmit a linear combination of all packets $s_1$ to $s_{V[n]}$, where the coefficients are chosen from a large enough field such that the coded combinations are independent with high probability. We refer to $V[n]$ as the transmit index in slot $n$. 
\end{defn} 
\begin{conj}
\label{conj:full_rank}
Since the packets are required in-order at the receiver, we believe that given transmit index $V[n]$, there is no loss of generality in including all packets $s_1$ to $s_{V[n]}$.
\end{conj}
Hence we believe that there is no loss of generality in restricting our attention to full-rank codes. % conjecture that the optimal scheme in this class is also optimal in general. %We further restrict $V[n]$ to be increasing because there is no feedback to the source about channel erasures. Only if there was feedback about past erasures, it could possibly be of advantage to reduce the transmit index $V[n]$ at some point and repeat packets that the feedback indicates to have been erased. 

\begin{thm}
\label{thm:opt_no_feedback}
The optimal throughput-delay trade-off among full-rank codes is $(\tau, \lambda) = (r , D(r|| p))$ for all $0 \leq r \leq p$. It is achieved by the coding scheme with $V[n] = \lceil rn \rceil$ for all $n$. 
\end{thm}
The term $D(r||p)$ is the binary information divergence function which is defined for $0 < p,r < 1$ as
\begin{equation}
D(r || p ) = r \log \frac{r}{p} + (1-r) \log \frac{1-r}{1-p}.
\end{equation}
Note that as $r \rightarrow 0$, $D(r||p)$ converges to $-\log(1-p)$, which is the optimal $\lambda$, as given by Lemma~\ref{lem:range_of_tau_lambda}. 

To prove Theorem~\ref{thm:opt_no_feedback}, we first show that the scheme with transmit index $V[n] = \lceil rn \rceil$ in time slot $n$ achieves the trade-off $(\tau, \lambda) = (r , D(r|| p)$. Then we prove the converse by showing that no other full-rank scheme gives a better trade-off. 

\begin{proof}[Proof of Achievability]
Consider the scheme with transmit index $V[n] = \lceil r n\rceil$, where $r$ represents the rate of adding new packets to the transmitted stream. The rate of adding packets is below the capacity of the erasure channel. Thus it is easy to see that the throughput $\tau = r$. Let $E[n]$ be the number of combinations, or equations received until time $n$. It follows the binomial distribution with parameter $p$. All packets  $s_1 \cdots s_{V[n]}$ are decoded when $E[n] \geq V[n]$. Define the event $G_n = \{ E[n] < V[n] \text{ for all } 1 \leq j \leq n\}$, that there is no packet decoding until slot $n$. The tail distribution of time $T$ between successive in-order decoding instants is,
\begin{align}
\Pr(T>n) &= \sum_{k=0}^{\lceil n r \rceil - 1}  \Pr(E[n]=k) \Pr(G_n | E[n] = k),\nonumber\\
&= \sum_{k=0}^{\lceil n r \rceil - 1}  \binom{n}{k} p^{k} (1-p)^{n-k} \Pr(G_n | E[n] = k), \nonumber \\
& \doteq \sum_{k=0}^{\lceil n r \rceil - 1}  \binom{n}{k} p^{k} (1-p)^{n-k},  \label{eqn:asym_T_no_fb_1}\\
& \doteq \binom{n}{\lceil n r \rceil - 1} p^{\lceil n r \rceil - 1} (1-p)^{n-\lceil n r \rceil + 1},  \label{eqn:asym_T_no_fb_2}\\
& \doteq e^{-n D(r||p)}, \label{eqn:asym_T_no_fb_3}
\end{align}
where in \eqref{eqn:asym_T_no_fb_1}, we remove the $\Pr(G_n | E[n] = k)$ when we take the asymptotic equality $\doteq$ because, by the Generalized Ballot theorem from \cite{durrett}, we can show that $\Pr(G_n | E[n] = k)$ is $\omega(1/n)$. Hence it is sub-exponential and does not affect the exponent of $\Pr(T> n)$. In \eqref{eqn:asym_T_no_fb_2}, we only retain the $k = \lceil n r \rceil - 1$ term from the summation because for $r \leq p$, that term asymptotically dominates other terms. Finally, we use the Stirlings approximation of the binomial coefficient $\binom{n}{k} \approx e^{n H(k/n)}$ to obtain \eqref{eqn:asym_T_no_fb_3}.

Hence we have proved that the scheme with $V[n] = \lceil rn \rceil$ achieves the throughput-delay trade-off $(\tau, \lambda) = (r , D(r|| p)$.
\end{proof}

\begin{proof}[Proof of Converse]
First let us show that the transmit index $V[n]$ of the optimal full-rank scheme should be non-decreasing in $n$. Given a scheme which does not satisfy the non-decreasing property, we can permute the order of transmitting the coded packets such that $V[n]$ is non-decreasing in $n$. Changing the order of the transmitted packets will not affect the throughput $\tau$. And it can in fact improve the in-order decoding exponent $\lambda$ because decoding can occur sooner when the initial coded packets include fewer source packets.

In the proposed scheme with $V[n] = \lceil rn \rceil$, we add new packets to the transmitted stream at a constant rate $r$. But in general a full-rank scheme can vary the rate of adding packets. Suppose it uses rate $r_i$ for $n_i$ slots for all $1 \leq i \leq L$, such that $\sum_{i=0}^{L} n_i = n$ and $\sum_{i=1}^{L} n_i r_i = nr$. Then, the tail distribution of time $T$ between successive in-order decoding instants is,
\begin{align}
\Pr(T>n) &= \sum_{k=0}^{\lceil \sum_{i=1}^{L} n_i r_i \rceil - 1}  \Pr(E[n]=k) \Pr(G_n | E[n] = k),\label{eqn:asym_T_no_fb_4}\\
&\doteq \sum_{k=0}^{\lceil n r \rceil - 1}  \binom{n}{k} p^{k} (1-p)^{n-k}, \label{eqn:asym_T_no_fb_5} \\
& \doteq e^{-n D(r||p)}. \label{eqn:asym_T_no_fb_6}
\end{align}
Varying the rate of adding packets affects the term $\Pr(G_n | E[n] = k)$ in \eqref{eqn:asym_T_no_fb_4}, but it is still $\omega(1/n)$ and we can eliminate it when we take the asymptotic equality in \eqref{eqn:asym_T_no_fb_5}. As a result, the in-order delay exponent is same as that if we had a constant rate $r$ of adding new packets to the transmitted stream. Hence we have proved that no other full-rank scheme can achieve a better $(\tau, \lambda)$ trade-off than $V[n] = \lceil nr \rceil$ for all $n$.
\end{proof}

Fig.~\ref{fig:inst_fb_no_fb_tradeoff} shows the $(\tau, \lambda)$ trade-off for the immediate feedback and no feedback cases, with success probability $p = 0.6$. The optimal trade-off with any feedback delay $d$ lies in between these two extreme cases.
\begin{figure}[t]
\centering
\includegraphics[width=3.5in]{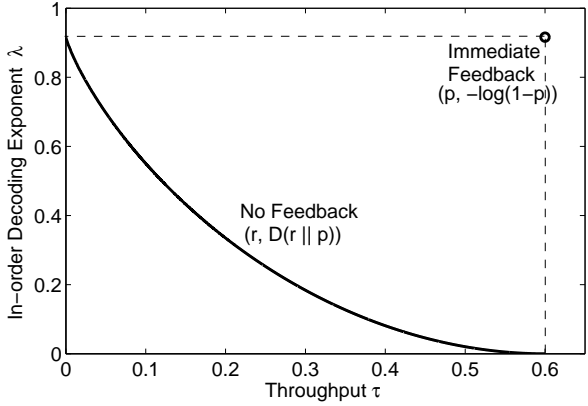}
\caption{The trade-off between in-order decoding exponent $\lambda$ and throughput $\tau$ with success probability $p = 0.6$ for the immediate feedback $(d= 1)$ and no feedback $(d = \infty)$ cases.}
\label{fig:inst_fb_no_fb_tradeoff}
\end{figure}

\section{General Block-wise Feedback}
\label{sec:block_wise_fb}
In Section~\ref{sec:immediate_feedback} and Section~\ref{sec:no_feedback} we considered the extreme cases of immediate feedback $(d=1)$ and no feedback $(d = \infty)$ respectively. We now analyze the $(\tau, \lambda)$ trade-off with general block-wise feedback delay of $d$ slots. We restrict our attention to a class of coding schemes called time-invariant schemes, which are defined as follows. 

\begin{defn}[Time-invariant schemes]
\label{defn:time_invariant}
A time-invariant scheme is represented by a vector $\mathbf{x} = [x_1, \cdots x_d]$ where $x_i$, for $ 1 \leq i \leq d$, are non-negative integers such that $\sum_{i} x_i = d$. In each block we transmit $x_i$ linear combinations of the $i$ lowest-index unseen packets in the stream. 
\end{defn}
The above class of schemes is referred to as time-invariant because the vector $\mathbf{x}$ is fixed across all blocks. Observe that as $d \rightarrow \infty$, the class of time-invariant schemes are equivalent to full-rank codes defined in Definition~\ref{defn:full_rank_codes}.
\begin{conj}
\label{conj:time_invar}
For any coding scheme, there exists a corresponding time-sharing policy between time-invariant schemes that gives the same or strictly better $(\tau, \lambda)$ trade-off. 
\end{conj}
We believe this conjecture is true because, it can be shown that any full-rank code can be expressed a time-sharing time-invariant scheme. By Conjecture~\ref{conj:full_rank} it follows that there is no loss of generality in focusing on time-invariant schemes.

There is also no loss of generality in restricting the length of the vector $\mathbf{x}$ to $d$. This is because we are still transmitting $d$ independent coded packets. And adding fewer source packets to the coded combinations, can only increase the exponent $\lambda$. % Hence the throughput $\tau$ is not reduced. And by having $x_1 \geq 1$, we are in fact improving the exponent $\lambda$, over the case where $x_1 =0$. %without giving increasing throughput $\tau$. %Similarly, there is no loss of generality in assuming $x_1 = 1$. %\textcolor{red}{Explain more here}

\subsection{Analyzing the $(\tau, \lambda)$ of time-invariant schemes}
Given a vector $\mathbf{x}$, define $p_d$, as the probability of decoding the first unseen packet during the block, and $S_d$ as the number of innovative coded packets that are received during that block. We can express $\tau_{\mathbf{x}}$ and $\lambda_{\mathbf{x}}$ in terms of $p_d$ and $S_d$ as,
\begin{align}
(\tau_{\mathbf{x}}, \lambda_{\mathbf{x}}) &= \left(\frac{\expec[S_d]}{d}, -\frac{1}{d}\log(1-p_d) \right),\label{eqn:lambda_d}
\end{align}
where we get throughput $\tau_{\mathbf{x}}$ by normalizing the $\expec[S_d]$ by the number of slots in the slots. We can show that the probability $\Pr(T>kd)$ of no in-order packet being decoded in $k$ blocks is equal $(1-p_d)^k$. Substituting this in \eqref{eqn:lambda_def} we get $\lambda_{\mathbf{x}}$.

\begin{exmp}
Consider the time-invariant scheme $\mathbf{x} = [1, 0, 3, 0]$ where block size $d=4$. That is, we transmit $1$ combination of the first unseen packet, and $3$ combinations of the first $3$ unseen packets. Fig.~\ref{fig:block_wise_exmp} illustrates this scheme for one channel realization. The probability $p_d$ and $\expec[S_d]$ are,
\begin{align}
p_d &= p + (1-p)\binom{3}{3} p^3 (1-p)^0 =  p + (1-p) p^3, \label{eqn:exmp_p_d}\\
\expec[S_d] &= \sum_{i=1}^{3} i \cdot \binom{4}{i} p^i (1-p)^{4-i} + 3p^4 = 4p-p^4, \label{eqn:exmp_E_S_d}
\end{align}
where in \eqref{eqn:exmp_E_S_d}, we get $i$ innovative packets if there are $i$ successful slots for $1 \leq i \leq 3$. But if all $4$ slots are successful we get only $3$ innovative packets. We can substitute \eqref{eqn:exmp_p_d} and \eqref{eqn:exmp_E_S_d} in \eqref{eqn:lambda_d} to get the $(\tau, \lambda)$ trade-off.
\end{exmp}

\begin{figure}[t]
\centering
\includegraphics[scale=0.45]{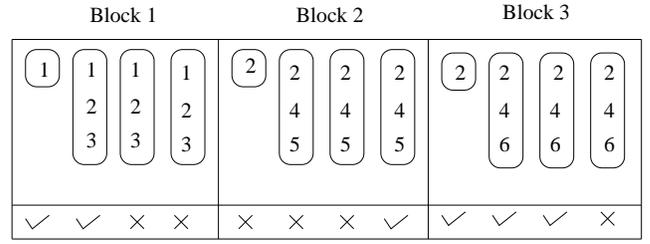}
\caption{Illustration of the time-invariant scheme $\mathbf{x} = [1,0,3,0]$ with block size $d=4$. Each bubble represents a coded combination, and the numbers inside it are the indices of the source packets included in that combination. The check and cross marks denote successful and erased slots respectively. The packets that are ``seen" in each block are not included in the coded packets in future blocks.}
\vspace{-0.1cm}
\label{fig:block_wise_exmp}
\end{figure}

\begin{rem}
\label{rem:uniqueness_of_schemes}
Time-invariant schemes with different $\mathbf{x}$ can be equivalent in terms of the $(\tau, \lambda)$. In general, given $x_1 \geq  1$, if any $x_i = 0$, and $x_{i+1} = x \geq 1$, then the scheme is equivalent to setting $x_i = 1$ and $x_{i+1} = x-1$, keeping all other elements of $\mathbf{x}$ the same. For example, $\mathbf{x} = [ 1, 1, 2,0]$ gives the same $(\tau, \lambda)$ as $\mathbf{x} = [ 1,0,3,0]$. 
\end{rem}
\subsection{Cost of Achieving Optimal $\tau$ or $\lambda$}
In Section~\ref{sec:immediate_feedback} we saw that for the immediate feedback case, we can achieve $(\tau, \lambda) = (p , -\log(1-p))$. However, when the feedback is delayed we can achieve optimal $\tau$ (or $\lambda$) only at the cost of sacrificing the optimality of the other metric. We now find the best achievable $\tau $ (or $\lambda$) with optimal $\lambda$ (or $\tau$).

\begin{lem}[Cost of Optimal Exponent $\lambda$]
\label{lem:cost_of_opt_lambda}
For a feedback delay of $d$ slots, the best achievable throughput is $\tau = (1-(1-p)^d)/d$, when the in-order decoding exponent $\lambda = -\log(1-p)$. 
\end{lem}
\begin{proof}
If we want to achieve $\lambda = -\log(1-p)$, we require $p_d$ in \eqref{eqn:lambda_d} to be equal to $1-(1-p)^d$. The only scheme that can achieve this is $\mathbf{x} = [d, 0, \cdots, 0]$, where we transmit $d$ copies of the first unseen packet. The number of innovative packets $S_d$ received in every block is $1$ with probability $1-(1-p)^d$, and zero otherwise. Hence, the best achievable throughput is $\tau = (1-(1-p)^d)/d$ with optimal $\lambda = -\log(1-p)$.
\end{proof}

This result gives us insight on how much bandwidth (which is proportional to $1/\tau$) is needed for a highly delay-sensitive application which needs $\lambda$ to be as large as possible. 

\begin{lem}[Cost of  Optimal Throughput $\tau$]
\label{lem:cost_of_opt_tau}
For a feedback delay of $d$ slots,  the best achievable in-order decoding exponent is $\lambda = -\log(1-p)/d$, when the throughput $\tau = p$. 
\end{lem}
\begin{proof}
If we want to achieve $\tau = p$, we need to guarantee an innovation packet in every successful slot. The only time invariant scheme that achieve this is $\mathbf{x} = [ 1,1, \cdots 1]$, and the vectors $\mathbf{x}$ that are equivalent to it as given by Remark~\ref{rem:uniqueness_of_schemes}. With $\mathbf{x} = [ 1,1, \cdots 1]$, the probability of decoding the first unseen packet is $p_d = p$. Substituting this in \eqref{eqn:lambda_d} we get $\lambda =-\nicefrac{\log(1-p)}{d}$, the best achievable $\lambda$ when $\tau = p$.
\end{proof}

Fig.~\ref{fig:cost_of_opt_tau_lambda} shows the best achievable $\tau$ and $\lambda$ versus $d$, when the other metric is at its optimal value. The plots in Fig.~\ref{fig:cost_of_opt_tau_lambda} correspond to moving leftwards and downwards respectively from the optimal trade-off $(p, -\log(1-p))$ in Fig.~\ref{fig:inst_fb_no_fb_tradeoff}.

%\textcolor{red}{Comment that we do not need to add more than $d$ undecoded packets into the stream in each block}
%
\begin{figure}[t]
\centering
\includegraphics[width=3.5in]{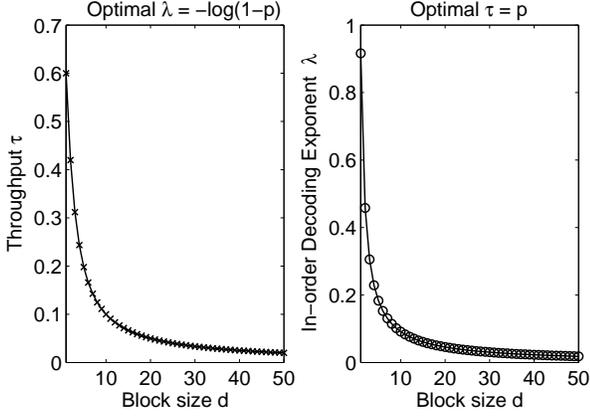}
\caption{Plot of the best achievable $\tau$ (or $\lambda)$ versus $d$, while maintaining the optimal value of the other metric $\lambda$ (or $\tau$), for channel success probability $p=0.6$. }
\label{fig:cost_of_opt_tau_lambda}
\vspace{-0.2cm}
\end{figure}

\subsection{Finding Optimal $(\tau, \lambda)$ Trade-off}
For any given throughput $\tau$, our aim is to find the transmission scheme that achieves the maximum $\lambda$. We first prove that any convex combination of achievable points $(\tau, \lambda)$ can be achieved. %For small $d$, we can exactly determine the set of time invariant schemes $\mathbf{x}$ whose convex combination gives the optimal trade-off. For general $d$, we derive the trade-off for a sub-class of time-invariant schemes, and show that it is close to the optimal trade-off.

\begin{thm}[Convex Combinations of Time-invariant Schemes]
\label{thm:interpolate_bw_schemes}
Given time-invariant schemes $\mathbf{x}^{(i)}$ for $1 \leq i \leq B$, we can achieve the throughput-delay trade-off given by any convex combination of the points $(\lambda_{\mathbf{x}^{(i)}}, \tau_{\mathbf{x}^{(i)}})$ by time-sharing between the schemes.  
\end{thm}
\begin{proof}
Here we prove the result for $B=2$, that is time-sharing between two schemes. It can be extended to general $B$ using induction. Given two time-invariant schemes $\mathbf{x}^{(1)}$ and $\mathbf{x}^{(2)}$ which achieve the throughput-delay trade-offs $(\lambda_{\mathbf{x}^{(1)}}, \tau_{\mathbf{x}^{(1)}})$ and $(\lambda_{\mathbf{x}^{(2)}}, \tau_{\mathbf{x}^{(2)}})$ respectively, consider a time-sharing strategy where, in each block we use the scheme $\mathbf{x}^{(1)}$ with probability $\mu$ and scheme $\mathbf{x}^{(2)}$ otherwise. Then, it is easy to see that the throughput on the new scheme is $\tau = \mu \tau_{\mathbf{x}^{(1)}} + (1-\mu) \tau_{\mathbf{x}^{(2)}}$. 

Now we prove the in-order decoding exponent $\lambda$ is also a convex combinations of $\lambda_{\mathbf{x}^{(1)}}$ and $\lambda_{\mathbf{x}^{(2)}}$. Let $p_{d_1}$ and $p_{d_2}$ be the probabilities of decoding the first unseen packet in a block using scheme $\mathbf{x}^{(1)}$ and $\mathbf{x}^{(2)}$ respectively. Suppose in an interval with $k$ blocks, we use scheme $\mathbf{x}^{(1)}$ for $h$ blocks, and scheme $\mathbf{x}^{(2)}$ in the remaining blocks, we have 
\begin{equation}
\Pr(T > kd) = (1-p_{d_1})^{h} (1-p_{d_2})^{k-h}.
\end{equation}
Using this we can evaluate $\lambda$ as,
\begin{align}
\lambda &= \lambda_{\mathbf{x}^{(1)}} \lim_{k \rightarrow \infty} \frac{h}{k} +\lambda_{\mathbf{x}^{(2)}} \lim_{k \rightarrow \infty} \frac{k-h}{k} \label{eqn:lambda_convex_comb} \\
 &= \mu \lambda_{\mathbf{x}^{(1)}} + (1-\mu) \lambda_{\mathbf{x}^{(2)}}
\end{align}
where we get \eqref{eqn:lambda_convex_comb} using \eqref{eqn:lambda_d}. As $k \rightarrow \infty$, by the weak law of large numbers, the fraction $h/k$ converges to $\mu$. Hence, we have shown that we can interpolate between the $(\tau, \lambda)$ trade-off of two policies by time-sharing between them.
\end{proof}

The main implication of Theorem~\ref{thm:interpolate_bw_schemes} is that, to find the optimal $(\tau,\lambda)$ trade-off, we only have to find the points $(\tau_{\mathbf{x}} ,\lambda_{\mathbf{x}})$ that lie on the convex envelope of the achievable region spanned by all possible $\mathbf{x}$. We determine this optimal trade-off for $d=2, 3$ in Lemma~\ref{lem:d_2_tradeoff} and Lemma~\ref{lem:d_3_tradeoff} below. 

\begin{lem}[Optimal Trade-off for $d=2$]
\label{lem:d_2_tradeoff}
The optimal $(\tau, \lambda)$ trade-off is the line joining points $\left( (1-(1-p)^2)/2, -\log(1-p)\right)$ and $\left(p, -\log(1-p)/2\right)$. 
\end{lem}
\begin{proof}
When $d=2$ there are only two possible time-invariant schemes  $\mathbf{x} = [2,0]$ and $[1, 1]$ that give unique $(\tau, \lambda)$. By Remark~\ref{rem:uniqueness_of_schemes}, all other valid vectors $\mathbf{x}$ are equivalent to one of these schemes. From Lemma~\ref{lem:cost_of_opt_lambda} and Lemma~\ref{lem:cost_of_opt_tau} we know that the $(\tau,\lambda)$ for these schemes are $( (1-(1-p)^2)/2, -\log(1-p))$ and $(p, -\log(1-p)/2)$ respectively. By Theorem~\ref{thm:interpolate_bw_schemes} we can achieve all $(\tau, \lambda)$ on the line joining these two points by time-sharing between the two policies. 
\end{proof}

\begin{lem}[Optimal Trade-off for $d=3$]
\label{lem:d_3_tradeoff}
The optimal $(\tau, \lambda)$ trade-off when $d=3$ is the piecewise linear curve joining points
\begin{align}
(\tau_A, \lambda_A) &= \left( \frac{1-(1-p)^3}{3}, -\log(1-p)\right), \label{eqn:tau_lambda_A} \\
(\tau_B, \lambda_B) &= \left( \frac{2p(2-p)}{3}, -\frac{2}{3}\log(1-p)\right), \label{eqn:tau_lambda_B}\\
(\tau_C, \lambda_C) &= \left(p, -\frac{\log(1-p)}{3}\right). \label{eqn:tau_lambda_C}
\end{align} 

\end{lem}
\begin{proof}
When $d=3$ there are four time-invariant schemes $\mathbf{x}^{(1)} = [3, 0, 0], \mathbf{x}^{(2)} = [2,1,0],\mathbf{x}^{(3)} =[1,2,0]$ and $\mathbf{x}^{(4)} = [1,1,1]$ that give unique $(\tau, \lambda)$, as given by Definition~\ref{defn:time_invariant} and Remark~\ref{rem:uniqueness_of_schemes}. From Lemma~\ref{lem:cost_of_opt_lambda} and Lemma~\ref{lem:cost_of_opt_tau} we know that $(\tau_{\mathbf{x}^{(1)}},\lambda_{\mathbf{x}^{(1)}}) =(\tau_A, \lambda_A)$ and $(\tau_{\mathbf{x}^{(4)}},\lambda_{\mathbf{x}^{(4)}}) =(\tau_C, \lambda_C)$. %These points will always be on the optimal trade-off because they correspond to the optimal $\lambda$ and optimal $\tau$ respectively. 

For the other two schemes, we first evaluate $p_d$ and $\expec[S_d]$ and substitute them in \eqref{eqn:lambda_d} to get, $(\tau_{\mathbf{x}^{(2)}},\lambda_{\mathbf{x}^{(2)}}) =(\tau_B, \lambda_B)$ and
\[ (\tau_{\mathbf{x}^{(3)}},\lambda_{\mathbf{x}^{(3)}}) =\left( (3p-p^3)/3, -(\log(1-p)^2(1+p))/3\right). \]

We can show that $\mathbf{x}^{(2)}$ gives a better trade-off than $\mathbf{x}^{(3)}$ by showing that for all $p$, the slopes of the lines joining $(\tau_{\mathbf{x}^{(i)}},\lambda_{\mathbf{x}^{(i)}})$ for $i = 1, \cdots 4$ satisfy,
%values of probability $p$, 
%\begin{itemize}
\begin{align}
\frac{\lambda_{\mathbf{x}^{(1)}} - \lambda_{\mathbf{x}^{(2)}}}{ \tau_{\mathbf{x}^{(1)}} - \tau_{\mathbf{x}^{(2)}}} &\geq \frac{\lambda_{\mathbf{x}^{(1)}} - \lambda_{\mathbf{x}^{(3)}}}{ \tau_{\mathbf{x}^{(1)}} - \tau_{\mathbf{x}^{(3)}}} \\
\frac{\lambda_{\mathbf{x}^{(2)}} - \lambda_{\mathbf{x}^{(4)}}}{ \tau_{\mathbf{x}^{(2)}} - \tau_{\mathbf{x}^{(4)}}} &\leq \frac{\lambda_{\mathbf{x}^{(3)}} - \lambda_{\mathbf{x}^{(4)}}}{ \tau_{\mathbf{x}^{(3)}} - \tau_{\mathbf{x}^{(4)}}}.
\end{align}
%
%\item Slope $\frac{\lambda_{\mathbf{x}^{(1)}} - \lambda_{\mathbf{x}^{(2)}}}{ \tau_{\mathbf{x}^{(1)}} - \tau_{\mathbf{x}^{(2)}}} \geq \frac{\lambda_{\mathbf{x}^{(1)}} - \lambda_{\mathbf{x}^{(3)}}}{ \tau_{\mathbf{x}^{(1)}} - \tau_{\mathbf{x}^{(3)}}}$ and,
%\item Slope $\frac{\lambda_{\mathbf{x}^{(2)}} - \lambda_{\mathbf{x}^{(4)}}}{ \tau_{\mathbf{x}^{(2)}} - \tau_{\mathbf{x}^{(4)}}} \leq \frac{\lambda_{\mathbf{x}^{(3)}} - \lambda_{\mathbf{x}^{(4)}}}{ \tau_{\mathbf{x}^{(3)}} - \tau_{\mathbf{x}^{(4)}}}$.
%\end{itemize}

\end{proof}

The trade-off for $d=2$ and $d=3$ with $p= 0.6$ is shown in Fig.~\ref{fig:block_wise_tradeoff}. The point below the piece-wise linear curve for $d=3$, corresponding to the sub-optimal scheme $\mathbf{x}^{(3)}=[1,2,0]$. We observe that the optimal trade-off becomes significantly worse are $d$ increases. From this we can imply that frequent feedback to the source is important in delay-sensitive applications to ensure fast in-order decoding of packets. %The $(\tau_{\mathbf{x}^{(i)}},\lambda_{\mathbf{x}^{(i)}} )$ for all $ 1 \leq i \leq 4$ for $p=0.6$ are shown in Fig.~\ref{fig:block_wise_tradeoff}. 

\begin{figure}[t]
\centering
\includegraphics[width=3.5in]{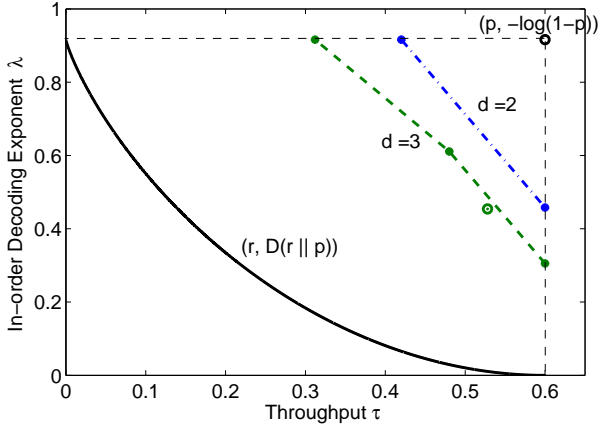}
\caption{The throughput-delay trade-off with $p=0.6$ for $d=2,3$ which can be shown to be optimal over all convex combinations of time-invariant schemes. The point just below the piece-wise linear curve for $d=3$, corresponding to the sub-optimal scheme $\mathbf{x} = [1, 2, 0]$.}
\label{fig:block_wise_tradeoff}
\vspace{-0.2cm}
\end{figure}
%

%\vspace{0.5cm}

% all the possible vectors $\mathbf{x}$ and find the set of $\mathbf{x}$ that form the optimal $(\tau, \lambda)$.
For general $d$, it is hard to search for the $(\tau_{\mathbf{x}}, \lambda_{\mathbf{x}})$ that lie on the optimal trade-off.
We suggest a set of time-invariant schemes which are easy to analyze and they give a good $(\tau, \lambda)$ trade-off. 

\begin{defn}[Suggested Schemes for General $d$]
\label{defn:suggested_schemes}
For general $d$ we suggest schemes with $x_1 = a$ and $x_{d-a+1} = d-a$, for $a = 1, \cdots d$. They give the throughput-delay trade-off
\begin{align}
(\tau,\lambda) &=\left( \frac{1-(1-p)^a + (d-a)p}{d}, -\frac{a}{d} \log (1-p) \right).\label{eqn:close_to_optimal}
\end{align}
\end{defn}

Fig.~\ref{fig:close_to_optimal_tradeoff} shows the trade-off given by \eqref{eqn:close_to_optimal}  for different values of $d$. Observe that for $d=2$ and $d=3$ the suggested schemes coincide with the optimal trade-off we derived in Lemma~\ref{lem:d_2_tradeoff} and Lemma~\ref{lem:d_3_tradeoff} and shown in Fig.~\ref{fig:block_wise_tradeoff}. As $d \rightarrow \infty$, and $a = \alpha d$, the trade-off converges to $( (1-\alpha)p, -\alpha \log(1-p))$ for $ 0 \leq \alpha \leq 1$, which is the line joining $(0, -\log(1-p))$ and $(p, 0)$. Numerical results suggest that for small $d$  this class of schemes gives the best trade-off among all possible time-invariant schemes $\mathbf{x}$, and close to optimal in general.

\begin{figure}[t]
\centering
\includegraphics[scale=0.45]{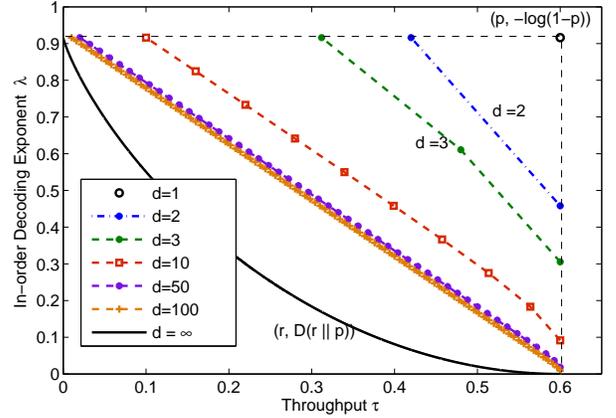}
\caption{The throughput-delay trade-off of the suggested coding schemes in Definition~\ref{defn:suggested_schemes} with $p=0.6$ and different values of feedback delay $d$. Numerical results suggest that this trade-off is optimal all convex combinations of time-invariant schemes for small $d$.}
\label{fig:close_to_optimal_tradeoff}
\vspace{-0.5cm}
\end{figure}
%
%\vspace{-0.5cm}
\section{Concluding Remarks}
\label{sec:conclu}
In this paper we analyze how block-wise feedback affects the trade-off between throughput $\tau$ and in-order decoding exponent $\lambda$, which measures the burstiness in-order packet decoding in streaming communication. When there is immediate feedback, we can simultaneously achieve the optimal $\tau$ and $\lambda$. But as the block size increases, and the frequency of feedback reduces, we have to compromise on at least one of these metrics. Our analysis gives us the insight that frequent feedback is crucial to ensure in-order packet delivery in delay-sensitive applications. 

Given that feedback comes in blocks of $d$ slots, we present a spectrum of coding schemes that span different points on the $(\tau, \lambda)$ trade-off as shown in Fig.~\ref{fig:close_to_optimal_tradeoff}. Depending upon the delay-sensitivity and bandwidth limitations of the applications, these codes provide the flexibility to choose a suitable operating point on trade-off. The proposed codes can be shown to be optimal over the broad class of full-rank codes for small feedback delay $d$, and when there is no feedback. %These codes are tunable such that we can achieve athroughput-delay trade-off that is suitable for the delay-sensitivity and bandwidth limitations of the application.  %Future directions include exploring the multicast scenario where there is a trade-off between throughput and delay, as well as between the different users that are sharing the channel. 

%In this paper we consider the problem of how to effectively utilize delayed feedback to the source to ensure in-order packet delivery to the user. We study the trade-off between throughput $\tau$ and in-order decoding exponent $\lambda$, which measures the burstiness in-order packet decoding. When there is immediate feedback, we can simultaneously achieve the optimal $\tau$ and $\lambda$. But as the frequency of feedback reduces, we have to compromise on at least one of these metrics. Given that feedback comes in blocks of $d$ slots, we present a spectrum of coding schemes that span different points on the $(\tau, \lambda)$ trade-off as shown in Fig.~\ref{fig:close_to_optimal_tradeoff}. The proposed schemes are optimal over the broad class of full-rank codes, for the no feedback and small feedback delay $d$ cases, and close to optimal in general. The proposed schemes can be used in wide variety of applications because we can choose an appropriate operating point on this trade-off, depending upon the delay-sensitivity, and bandwidth limitations of the application. Future directions include exploring the multicast scenario where there is a trade-off between throughput and delay, as well as between the different users that are sharing the channel. 

%========================== Bibliography  ========================%

\bibliographystyle{ieeetr}

\end{document}